\documentclass[12pt]{article}

\usepackage{graphicx,amssymb,amsmath,amsthm,cite,cases}
\usepackage{algorithmic,setspace}
\usepackage[ruled]{algorithm} 
\usepackage[draft]{hyperref}
\usepackage{mathrsfs}

\def\bel{\begin{linenomath}}
\def\eel{\end{linenomath}}

\def\be{\begin{equation}}
\def\ee{\end{equation}}
\def\ben{\begin{eqnarray}}
\def\een{\end{eqnarray}}

\def\ds{\displaystyle}
\def\expn{e^{\imath \frac{\pi(n-1)}{M}}}

\def\D{\mathcal{D}}

\def\R{\mathbb{R}}
\def\N{\mathbb{N}}
\def\V{\mathbb{V}}
\def\C{\mathbb{C}}

\def\B{\mathcal{B}}

\def\Re{\operatorname{Re}}
\def\Im{\operatorname{Im}}
\def\im{\operatorname{\imath}}

\def\jdqd{\jd{\cqd}}
\def\qs{{q^\star}}
\def\qd{{q^\diamond}}

\def\jd{j^\diamond}

\def\jdq{\jd\cq}

\def\cq{\{q\}}
\def\cs{\{s\}}
\def\cqs{\{\qs\}}
\def\cqd{\{\qd\}}

\def\til{\tilde}
\def\wtil{\widetilde}

\def\vP{\mathbf{p}}

\def\vt{\mathbf{t}}

\def\vR{\mathbf{r}}
\def\vRq{{\vR}{\cq}}
\def\vd{\mathbf{d}}

\def\vf{\mathbf{f}}

\def\vfa{\vf^{\rm{a}}}
\def\vfad{\vf^{\rm{a}}_\downarrow}
\def\vfaq{\vfa{\cq}}
\def\vfq{{\vf}{\cq}}
\def\sq{s\cq}

\def\ellq{{\ell}{\cq}}

\def\vg{\mathbf{g}}
\def\vy{\mathbf{y}}
\def\vft{\til{\mathbf{f}}}
\def\vgt{\til{\mathbf{g}}}

\def\vB{\mathbf{b}}
\def\vBq{{\vB}{\cq}}
\def\vW{\mathbf{w}}
\def\vWt{\mathbf{\wtil{w}}}
\def\vWq{{\vW}{\cq}}
\def\vWtq{{\vWt}{\cq}}
\def\vRe{\mathbf{r}}

\def\vd{\mathbf{d}}

\def\op{\hat{\operatorname{P}}}
\def\Join{\hat{\operatorname{J}}}

\def\kq{k(q)}
\def\kqs{k(\qs)}
\def\Nq{N_{\rm{b}}}
\def\Ns{N_{\rm{s}}}

\newcommand{\La}{ \left \langle}
\newcommand{\ra}{\Ra}
\newcommand{\la}{\La}
\newcommand{\Ra}{ \right \rangle}

\newcommand{\argmax}{\operatorname*{arg\,max}}
\newcommand{\argmin}{\operatorname*{arg\,min}}

\newtheorem{proposition}{Proposition}

\def\rc{\rm{c}}
\def\rs{\rm{s}}
\def\rcs{\rm{cs}}

\def\Dc{\D^{\rm{c}}}
\def\Ds{\D^{\rm{s}}}
\def\Dcs{\D^{\rm{cs}}}

\def\Bc{\B^{\rm{c}}}
\def\Bs{\B^{\rm{s}}}
\def\Bcs{\B^{\rm{cs}}}

\def\Dct{\D^{\rc 2}}
\def\Dst{\D^{\rs 2}}
\def\Dcst{\D^{\rcs 2}}

\def\Dcf{\D^{\rc 4}}
\def\Dsf{\D^{\rs 4}}
\def\Dcsf{\D^{\rcs 4}}

\newcommand{\Spann}{{\mbox{\rm{span}}}}

\title{Cooperative Greedy Pursuit 
Strategies for Sparse Signal Representation by Partitioning}

\author{Laura Rebollo-Neira\\
Mathematics Department\\
Aston University\\
B3 7ET, Birmingham, UK}

\begin{document}
\maketitle 
\baselineskip = 2\baselineskip

\begin{abstract}
Cooperative Greedy Pursuit Strategies are 
considered for approximating a signal partition 
subjected to a global constraint on sparsity. 
The approach aims at producing a high quality 
sparse approximation of the whole signal, using highly 
coherent redundant dictionaries. The cooperation takes
place by ranking the partition units for their sequential 
stepwise approximation, and is realized by means of
i)forward steps for the upgrading
of an approximation and/or ii) backward steps
for the corresponding downgrading. The advantage of 
 the strategy is illustrated by  
approximation of music signals using redundant 
trigonometric dictionaries. In addition to rendering 
 stunning improvements in sparsity 
with respect to the concomitant trigonometric 
basis, these dictionaries enable a fast implementation 
of the approach via the Fast Fourier Transform.

\end{abstract}



\section{Introduction}
Compressible signals, such as audio and vision data, 
are characterized by samples containing a good deal of 
redundancy.
Consequently,
if properly processed, the signal information content can be 
accessed from a reduced data set. 
Transformations for  
data reduction are said to produce a 
{\em{sparse representation}} of a 
signal if they can accurately reproduce the information the 
signal conveys,   
with significantly less points that those by which the
 original signal is given.   
 Most popular transformations for signal processing
(e.g. Discrete Cosine Transform and Discrete Wavelet Transform)
do not modify the signal size. A sparse representation is 
achieved, 
a posteriori, by disregarding the least relevant 
 points in the transformed domain. 
Nonetheless, 
 much higher sparsity in a signal  representation may be 
achieved  
by allowing for the expansion of the transformed domain, 
  thereby leaving room for dedicated transformations adapted  to the particular signal. 
Such a framework involves a large redundant set called `dictionary'. The aim is to 
 represent a 
signal as a superposition of, say $K$, dictionary's 
elements,  which are called `atoms', with $K$  much 
 less than  the signal size.  

Given a redundant dictionary, the problem of 
finding the sparsest approximation of a  signal, up to 
some predetermined error, is an NP-hard problem \cite{Nat95}.
Consequently, the interest in practical applications 
lies in the finding of those solutions which, 
being sparse in relation to 
other possible representations, are also easy and fast to 
construct. Sparse practical solutions can be found 
by what are known as {\em{Pursuit Strategies}}. Here the
discrimination between two broad categories is in order:
i)The Basis Pursuit based approaches, which  endeavor to 
obtain a sparse solution 
by  minimization of the 1-norm \cite{CDS01}. 
ii)Greedy algorithms which look for a sparse solution by stepwise selection of dictionary's atoms.  Practical
greedy algorithms, which originated as regression techniques 
in statistics \cite{FS81}, have been popularized in 
signal processing applications as  
Matching Pursuit (MP) \cite{MZ93}
and Orthogonal Matching Pursuit (OMP) \cite{PRK93} methods.
The approach, which in principle consider the 
stepwise selection of 
single atoms, has been extended to multiple atom
selection \cite{EKB10}. Dedicated  algorithms such as 
{{Stagewise Orthogonal Matching Pursuit}}\cite{DTD06}, 
 {{Compressive Sampling Matching Pursuit}}\cite{NT09}, 
and {{Regularized Orthogonal
Matching Pursuit}}\cite{NV10} are  known to be effective 
within the context of the 
 emerging theory of sampling, 
called {\em{compressive sensing$/$sampling}}. This 
theory
 asserts that sparsity of a representation may also
lead to more economical
data collection \cite{Don06,CRT06,Bar07,CW08,Bar11}.  
In that context 
the reconstruction problem is of a very particular 
nature, though: It is assumed that a signal is sparse in 
 an orthogonal basis and the goal is to reconstruct 
the signal from a reduced number of measures. 
 On the contrary,  in this Communication we 
 address the traditional representation matter: 
the signal is assumed to be completely
given by its samples. The aim is to produce a 
high quality approximation of all those samples, 
as a $K$-term superposition of atoms belonging to a 
highly coherent dictionary. In this case,
minimization of the 1-norm is not effective and 
step wise greedy selection of single atoms benefits 
sparsity results. 

In practice, when trying to approximate 
real life signals using a redundant
dictionary, there is a need to approximate  
by partitioning. While this requirement normally comes 
from storage and computational complexity demands,
it does not represent a 
disadvantage. On the contrary, to capture local properties 
of signals such as images, or music, non-local orthogonal 
transforms are also applied on partitions
to obtain superior results. The broadly used 
image compression standard JPEG and the music compression 
 standard MP3, for instance, 
 operate by partitioning as a first step in 
the compression process.

The central aim of this paper is to tailor
pursuit algorithms for approximation by partitioning 
 without significant increment in  computational complexity, 
in comparison with standard applications. 
The examples presented here clearly show
that a global constraint on sparsity, rather than quality 
constraints on the individual partition's units, 
may benefit enormously the quality of the signal 
approximation. 
This is true even if
the actual approximation of each unit
is performed individually. 
Because considerations
are restricted to producing high quality approximations,
there is no risk that blocking artifacts could          
appear due to the signal division into non-overlapping 
 pieces.                                 

\subsection*{Notational convention}
Throughout the paper $\R$, $\C$ and $\N$ indicate 
the sets of real, complex, and natural numbers, respectively.
Boldface letters are used to indicate Euclidean vectors,
whilst standard mathematical fonts indicate components,
e.g., $\vf \in \R^N,\, N \in \N$ is a vector of components
$f(i),\, i=1,\ldots,N$.
A partition of a signal $\vf \in \R^N$ is represented
as a set of disjoint pieces $\vfq \in \R^{\Nq},\, 
q=1,\ldots,Q$, which for simplicity are assumed to
be all of the same size and such that $Q \Nq =N$. 
The signal is reconstructed from the partition 
through the operation 
 $\vf = \Join_{q=1}^{Q} \vfq$ where $\Join$ represents the
 operator which concatenates the partition. The  
operation is defined as follows: given $\vf\{1\} \in 
\R^{\Nq}$ and $\vf\{2\} \in \R^{\Nq}$,  
the vector $\vf= \vf\{1\}\, \Join \,\vf\{2\}$ is a vector 
in $\R^{2 \Nq}$
having components $f(i)=f\{1\}(i)$ for $i=1,\ldots,\Nq$, 
and $f(i)=f\{2\}(i -\Nq)$ for $i=\Nq+1,\ldots,2\Nq$. Thus
$\vf = \Join_{q=1}^{Q} \vfq$ is a vector in $\R^{Q \Nq}$
having components 
$f(i)=f\cq(i-(q-1)\Nq),\, i=(q-1)\Nq+1,\ldots,q\Nq,\,q=1,\ldots,Q$.
Consequently $\La \vf, \vf\Ra = \| \vf \|^2= \sum_{q=1}^Q \left \| \vfq \right \|^2,$  were $\ds{\La \cdot,\cdot \Ra}$ 
 indicates the Euclidean inner product and 
$\| \cdot\|$ the induced 2-norm.
\subsection*{Paper contributions}
Given a signal partition 
$\vfq \in \R^{\Nq},\, q=1,\ldots,Q$
and a 
dictionary 
${\D=\left\{\vd_n \in \R^{\Nq}\,; \|\vd_n\|=1 \right\}_{n=1}^M}$
to approximate the 
elements $\vfq $ in 
 the partition, the
following outcome has been recently reported
\cite{RNMB13}:
 A very significant gain in the sparsity of a signal 
approximation 
may be effectively obtained by a greedy pursuit
strategy,
if the approximation of each piece $\vfq$ (called
`$q$-block') is
accomplished in a hierarchized manner. Suppose that
the $\kq$-term approximation of $\vfq$ is the
atomic decomposition:
$$\vfa\cq= \sum_{n=1}^{\kq} c\cq_n 
\vd_{\ell\cq_n},\quad q=1,\ldots,Q,$$
with the atoms $\vd_{\ell\cq_n},\,n=1,\ldots,\kq$
selected from the dictionary $\D$, via a stepwise
greedy pursuit strategy. Suppose also that the number
of atoms to approximate the whole signal $\vf$ is a
fixed value $K$, i.e., $K=\sum_{q=1}^Q \kq$. 
A possibility to handle this constraint is to
consider a hierarchized selection
 of the pieces $\vfq$ to be approximated in each
approximation step.
Some remarkable results of this strategy, which has
been termed Hierarchized Block Wise (HBW)
greedy strategy, are illustrated in \cite{RNMB13}
by approximating images using the
greedy algorithms
MP and OMP.  When these methods are
applied in the proposed HBW fashion
are called HBW-MP and HBW-OMP, respectively.

While \cite{RNMB13} focusses on highlighting the
suitability of the HBW-OMP$/$MP method 
 when the image approximation is  carried 
out in the wavelet domain,
this Communication extends the method 
 as well as the range of applicability. The
extended techniques are shown to produce hugely sparse 
high quality approximation of music signals. This   
is accomplished with trigonometric dictionaries, 
which are endowed with the
additional advantage of enhancing the competitiveness
of the approach in terms of computational complexity.

The extension of the idea outlined in \cite{RNMB13}
comprises:

\begin{itemize}
\item
A revised version of the approach, which optimizes 
the ranking of 
the pieces $\vfq,\,q=1,\ldots,Q$ for their step wise
approximation without significant extra computational cost. 
Additionally, the selection of atoms is also optimized
by including the Optimized Orthogonal
Matching Pursuit (OOMP) criterion \cite{RNL02}.
\item
A HBW backward approach for downgrading the approximation
when required. This is realized in a stepwise optimized 
manner, by removing terms in the approximations of the
selected blocks.
\item
An alternative to the HBW strategy which
consists of two stages.
The first stage involves the approximation
of the pieces $\vfq,\, q=1,\ldots,Q$, up to a
 tolerance error.
The second stage
refines the previous approximation by making possible
 the downgrading of the atomic decomposition of some blocks,
called  `donors',  and the upgrading    
of the atomic decompositions of another blocks, called
 `receivers'. Since this process is inspired in the
Swapping-based-Refinement of greedy strategies introduced in
\cite{ARN06}, we refer to it as 
HBW-Swapping-based-Refinement (HBW-SbR).  
\end{itemize}
Further contributions of the paper are
\begin{itemize}
\item
The illustration of the {\underline{huge}} gain in 
the sparsity of the representation  of 
music signals obtainable by the proposed approach when using
trigonometric dictionaries. 
\item
The finding that a mixed trigonometric dictionary 
having both, discrete cosine and sine components, renders 
the best sparsity performance for melodic music signals, 
in comparison with the single discrete  cosine (or sine)
 dictionaries with the same redundancy. 
This result contrasts with the fact that, the broadly 
used single  discrete cosine basis is the basis
  yielding best sparsity results for music 
signals, in comparison to other trigonometric basis 
of the same nature.\\
Note: The  discrete  cosine/sine
components of the mixed dictionary are not just 
the real and imaginary parts of the complex 
exponentials in a Fourier series. The involved 
phase factors make the whole difference in achieving 
high sparsity.
\item
The provision of implementation details in the form 
of algorithms for the proposed strategies and
the provision of dedicated algorithms to 
operate with trigonometric dictionaries via the Fast 
 Fourier Transform (FFT). 
\end{itemize}
A library of MATLAB functions for implementing
the proposed methods, and running  the
numerical examples in this paper,
 are available for downloading
on \cite{webpage}.

\section{Hierarchized Blockwise OMP, revised}
\label{HBW}
As already stated, a
signal $\vf \in \R^N$  will
be considered to be the composition of $Q$ identical 
and disjoint blocks 
$\vf = \Join_{q=1}^{Q} \vfq,$,
where $\vfq \in \R^{\Nq}$ 
 assuming that $\Nq Q=N$.
Each piece $\vfq \in \R^{\Nq}$ is approximated using 
 a dictionary $\D=\left\{\vd_n \in \R^{\Nq}\, 
; \| \vd_n\|=1 \right\}_{n=1}^M$, with 
$M>\Nq$, by an atomic decomposition 
of the form:
\be
\label{atoq}
\vfa\cq= \sum_{n=1}^{\kq}
 c\cq_n \vd_{\ell\cq_n},
\quad q=1\ldots,Q.
\ee
where the atoms $\vd_{\ell\cq_n},\, n=1,\ldots,\kq$  
 are  selected from the dictionary 
to approximate the block $q$. Thus, each
set of selected indices 
$\Gamma\cq=\{\ell\cq_n\}_{n=1}^{\kq}$
is a subset of the set of labels $\{n\}_{n=1}^M$
 which identify the atoms in $\D$. 
In general $\Gamma \cq \ne 
\Gamma \{p\}$ and $\kq \ne k(p)$ if $p\ne q$.

The HBW-OMP approach outlined in \cite{RNMB13} 
selects atoms in \eqref{atoq} as indicated by the 
OMP approach, i.e.: 
On setting initially $\kq=0$ and 
$\vR^{0}\cq=\vfq$,
 the algorithm picks the atoms in the
  atomic decomposition
of block $q$ by selecting one by one the indices 
which satisfy:
\be
\begin{split}
\ell{\cq}_{\kq+1}=
\operatorname*{arg\,max}_{n=1,\ldots,M}
 \left |\La \vd_n, \vR^{\kq}\cq  \Ra \right |,\, \text{with} \,\,
 \vR^{\kq}\cq = \vfq - \sum_{n=1}^{\kq} c{\cq}_n 
\vd_{\ellq_n}\,\,\text{if}\,\, \kq>0.
\end{split}
\label{omp}
\ee
The coefficients 
$c{\cq}_n,\,n=1,\ldots,\kq$ in \eqref{omp}
are such that the norm ${\left\|\vR^{\kq}\cq \right\|}$  is
 minimized.
This is ensured by requesting that 
$\vR^{\kq}\cq= \vfq- \op_{\V_{\kq}^q} \vfq$,
where $\op_{\V_{\kq}^q}$ is the orthogonal
projection operator
onto $\V_{\kq}^q=\Spann \left\{\vd_{\ell\cq_n} \right\}_{n=1}^{\kq}$. This condition also guarantees that the 
set $\left\{\vd_{\ell\cq_n} \right\}_{n=1}^{\kq}$ is 
linearly independent.

An implementation as in \cite{RNL02} 
provides us with two representations of $\op_{\V_{\kq}^q}$.
One of the 
representations is achieved by orthonormalization 
of the set $\left\{\vd_{\ell\cq_n} \right\}_{n=1}^{\kq}$.
The other by biorthogonalization of the same set.
We implement the orthogonalization by 
 Gram Schmidt method, including  
a re-orthogonalization step, as follows:  
 The orthonormal set
$\left\{\vWtq_n \right\}_{n=1}^{\kq}$ is inductively 
constructed from $\vWtq_1= \vWq_1=\vd_{\ellq_1}$
 through the process:
\be
\label{GS0}
\vWtq_n=\frac{\vWq_{n}}{\|\vWq_n\|},\quad \text{with}
\quad
\vWq_n= \vd_{\ell\cq_{n}} - \sum_{i=1}^{n-1}
{\vWtq_i}\La {\vWtq_i,\vd_{\ell\cq_{n}}} \Ra.
\ee
For numerical accuracy
at least one re-orthogonalization step
is usually needed. This
implies to recalculate the above vectors as
\be
\label{GSR}
\vWq_{n} \leftarrow \vWq_{n}- \sum_{i=1}^{n-1}
\vWtq_{i}
\La {\vWtq_{i},\vWq_{n}}\Ra,
\ee
with the corresponding normalization giving rise to the
orthonormal set $\left\{\vWtq_n \right\}_{n=1}^{\kq}$.
Notice that, whilst this set can be used to calculate  
the orthogonal projection of $\vfq$ onto $\V_{\kq}^q$ 
as 
\be
\label{pow}
\op_{\V_{\kq}^q} \vfq= \sum_{n=1}^{\kq} \vWtq_n \La {\vWtq_n, \vfq} \Ra,
\ee
this superposition is not the atomic
decomposition in terms of the selected atoms.
In order to produce such a decomposition we use the 
other representation for an orthogonal 
projector given in \cite{RNL02}. 
Namely, 
\be
\label{pob}
\op_{\V_{\kq}^q} \vfq= \sum_{n=1}^{\kq} 
\vd_{\ellq_n} \la{ \vBq_n,
\vfq }\ra.
\ee
For a fixed $q$  
the vectors $\vBq_n,\,n=1,\ldots,\kq$  in 
\eqref{pob} are biorthogonal to the selected atoms,
i.e. $\la \vd_{\ell\cq_n},\vBq_{m}\ra = 0$ if $n\ne m$ and 1 if
$n=m$, and
 span the identical subspace i.e., 
$\V_{\kq}^q= \Spann\{\vBq_{n}\}_{n=1}^{\kq}=\Spann\{\vd_{\ell\cq_n}\}_{n=1}^{\kq}$.  In order to fulfill the last 
condition {\em{all}} the vectors need to be updated
when a new atom is introduced in the spanning set.
Starting from $\vBq_1=\vd_{\ellq_1}$ both 
the calculation and upgrading of the biorthogonal set 
is attained recursively as 
follows
\be
\label{vB}
\begin{split}
\vBq_{\kq}&= \frac{\vWq_{\kq}}{\|\vWq_{\kq}\|^2},\,
\quad \text{with} \quad \vWq_{\kq}\,\,\text {as in \eqref{GS0}} \\ 
\vBq_n& \leftarrow \vBq_n - \vBq_{\kq}\la {\vd_{\ellq_{\kq}},\vBq_{n}}\ra,\quad n=1,\ldots,(\kq-1).
\end{split}
\ee
These vectors produce the required atomic decomposition 
for the block $q$.  The
coefficients in \eqref{atoq} are the
inner products in \eqref{pob}, i.e. $c\cq_n= 
\La {\vBq_{n}, \vfq} \Ra,\,n=1,\ldots,\kq.$ 

The HBW variant of a greedy strategy establishes an order 
for upgrading the
atomic decomposition of the blocks in the partition.
 Instead of completing the approximation of each block 
at once, 
 the atomic decomposition to be upgraded at each iteration 
corresponds to a block, say  $q^\star$, selected 
according to a greedy criterion.
In other words: the HBW version of a pursuit method 
for approximating a signal by partitioning involves
two selection instances: a) the selection of
dictionary's atoms for approximating each of the blocks 
in the partition and
b)the selection, at each iteration step,
of the block where the
upgrading of the approximation is to be realized.
In \cite{RNMB13} the blocks are selected through 
the process
\be
\label{hbwomp}
q^\star=
\operatorname*{arg\,max}_{q=1,\ldots,Q} 
\left |\La{\vd_{\ell\cq_{\kq+1}}, \vR^{\kq}\cq} \Ra \right|.
\ee
We notice, however, that condition \eqref{hbwomp} 
for the block's selection 
can be optimized, in a stepwise sense,
without significant computational cost. 
 The next proposition 
revise the condition \eqref{hbwomp} considered 
 in \cite{RNMB13}.
\begin{proposition}
Let $\ell{\cq}_{\kq+1}$
be the index arising, for each value of $q$,  from the
 maximization process \eqref{omp}.
In order to minimize the square norm of the total
residual $\|\vR^{k+1}\|^2$ at iteration $k+1$ the atomic
decomposition to be upgraded
should correspond to the block $q^\star$ such that
\be
\label{hbwoomp}
q^\star= \operatorname*{arg\,max}_{q=1,\ldots,Q}
\chi(q),\,\,\, {\text{where}}\,\,\, \chi(q)=
\left|\la {\vWtq_{\kq+1},\vfq} \ra\right|=
 \frac{\left|\La{\vd_{\ellq_{\kq+1}}, \vR^{\kq}\cq} \Ra \right|}{\left \|\vWq_{\kq+1} \right \|},
\ee
with $\vWq_{\kq+1}$ and $\,\vWtq_{\kq+1}$ 
as given in \eqref{GS0}.
\end{proposition}
\begin{proof}
Since at iteration $k+1$ the atomic decomposition of only
one block is upgraded by one atom, the total
residue at iteration ${k+1}$ is constructed as
$$\vR^{k+1}=\Join_{\substack{p=1\\ p \ne q}}^Q\,\vR\{p\}^{k(p)}\, \Join \,\vRq^{\kq+1}.$$
Then, $$\|\vR^{k+1} \|^2= \sum_{\substack{p=1\\ p \ne q}}^Q
\|\vR\{p\}^{k(p)}\|^2 + \|\vRq^{\kq+1}\|^2.$$
Moreover, since
$\vR^{\kq+1}\cq= \vfq - \op_{\V_{\kq+1}^q} \vfq$  using
 the orthonormal vectors \eqref{GS0} to calculate
$\op_{\V_{\kq+1}^q} \vfq$ we have:
$$\|\vR^{\kq+1}\cq\|^2= \|\vR^{\kq}\cq\|^2 - \left|\la {\vWtq_{\kq+1}}, \vfq \ra \right|^2,$$
which is minimum for the $q^\star$-value corresponding to 
the maximum value of $\left|\la {\vWtq_{\kq+1}},\vfq 
\ra\right|$.  Since $\vWq_{\kq+1}= \vd_{\ell\cq_{\kq+1}} -
\op_{\V_{\kq}^q} \vd_{\ell\cq_{\kq+1}}$ and $\op_{\V_{\kq}^q}$ is hermitian \eqref{hbwoomp} follows.
\end{proof}
Notice that, since the vectors $\vWtq_{\kq+1}$ will be
used for subsequent approximations,
the optimized ranking of blocks \eqref{hbwoomp}
 does not require significant extra computational effort.
Only $Q-1$ of these vectors will not have been
otherwise used
when the algorithm stops. Apart from that, the 
complexity of selecting the blocks remains being 
that of finding the maximum element of an array of
length $Q$,  i.e. O($Q$). 

Before presenting the algorithm details we would 
like to consider also the optimization of the 
selection of atoms (stage a) of a HBW pursuit strategy).
Within our implementation the optimized selection of atoms
is readily achievable 
by the Optimized 
Orthogonal Matching Pursuit (OOMP) approach \cite{RNL02},
which selects the index $\ell\cq_{\kq+1}$
through the maximization process:
\be
\ell\cq_{\kq+1}=
\operatorname*{arg\,max}_{\substack{n=1,\ldots,M}}
\frac{\left|\La{\vd_n, \vR^{\kq}\cq}\Ra\right|}{(1- s\cq_n)^{\frac{1}{2}}},\quad\text{with}\quad s\cq_n=\sum_{i=1}^{\kq}|\la \vd_n, \vWtq_{i}\ra |^2.
\label{oomp2}
\ee
The OOMP selection criterion \eqref{oomp2} minimizes, 
in a stepwise sense, 
the norm of the local error when approximating a 
single block \cite{RNL02}.
The extra computational cost, in comparison to the
OMP selection step (c.f. \eqref{omp}) 
is the calculation of the sum $s\cq_n$
 in the denominator of \eqref{oomp2}, for every atom in the 
dictionary, $n=1,\ldots,M$.  
As becomes clear in Algorithm \ref{alsn}, 
by saving the sum of previous iterations, at
iteration $\kq+1$ only the inner products
$\la \vd_n,  \vWtq_{\kq}\ra,\,
n=1,\ldots,M$ need to be calculated.
\subsection{Algorithmic Implementation}
The implementation details of the HBW-OOMP
 method 
is given in Algorithm 1, which is 
realized through 
Algorithms \ref{alell} and \ref{alsn}. 
For a fixed number $K$ the
algorithm iterates until the condition
$\sum_{q=1}^Q \kq=K$ is met.
\newcounter{myalg}
\begin{algorithm}[!htp]
\refstepcounter{myalg}
\setstretch{1.25}
\caption{HBW-OOMP Procedure}
\label{alhbwoomp}
\begin{algorithmic}
\STATE{{\bf{Input:}}\, Signal partition 
$\vfq \in \R^{\Nq},\, q=1,\ldots,Q$.
Dictionary $\D=\{\vd_n \in \R^{\Nq}\,; \|\vd_n\|=1\}_{n=1}^M$. Number $K$ of total atoms to approximate 
the whole signal.}\\

\STATE{{\bf{Output:}}\, Sets $\Gamma\cq, \,q=1,\ldots,Q$ 
containing the indices of the selected atom for each 
block. 
Orthonormal and biorthogonal 
sets  $ \ds{\{\vWtq_n\}_{n=1}^{\kq}}$  and
$\ds{\{\vBq_n\}_{n=1}^{\kq}}$, for each block 
(c.f. \eqref{GSR} and \eqref{vB} respectively).
Coefficients of the atomic decompositions for 
each block $\ds{\{c\cq_n\}_{n=1}^{\kq}},\, q=1,\ldots,Q$. 
Approximation of the partition, $\vf^{{\rm{a}}}\cq,\, q=1,\ldots,Q$.
Approximation of the  whole signal $\vf^{{\rm{a}}}$.}\\

\STATE {\COMMENT {Initialization}}

\STATE {$j=0$}

\FOR {$q=1:Q$} 
\STATE{$k(q)=1;\, \vRq= \vfq; \,\sq_n=0,\,n=1,\ldots,M$}

\COMMENT {Select the potential first atom for
 each block, as below} 

\STATE{$\ds{\ellq=\argmax_{n=1,\ldots,M}\left|\la \vd_n,\vfq \ra \right|}$}
\STATE{
$\ds{\chi(q)=\left|\la \vd_{\ellq},\vfq\ra \right|}\,\,$ \COMMENT{Store the maximum for each block}}

\STATE{ Set $\ds{\vWq_1=\vd_{\ellq};\, \vWtq_1=\vWq_1;\,  \vBq_1=\vWtq_1; \, \Gamma\cq=\emptyset}$ } 

\ENDFOR

\WHILE {$j < K$} 
\STATE{
\COMMENT{Select the block to be approximated 
(c.f.\eqref{hbwoomp}) and store the index of the atom in the approximation of that block, as below}}

\STATE{$\ds{\qs= \argmax_{q=1,\ldots,Q} \chi(q)}$}

\STATE{$\Gamma{\cqs} \leftarrow \Gamma{\cqs} \cup \ell\cqs$}

\COMMENT{Update the residue corresponding to block $\qs$, as 
below}

\STATE{$\ds{\vR\cqs \leftarrow  \vR\cqs-\vWt\{\qs\}_{\kqs}\la \vWt\cqs_{\kqs}, \vf\cqs \ra}$
}

\IF {$\kqs>1$}
\STATE{
\COMMENT{Upgrade the biorthogonal 
 set $\ds{\{\vB\cqs_n\}_{n=1}^{\kqs}}$  to include the 
 atom $\ds{\vd_{\ell\cqs}}$ (c.f. \eqref{vB})}}
\ENDIF


\STATE{\COMMENT{Apply  
Algorithm \ref{alell}
to select the index ${\ell\cqs}$ for a new potential 
atom for the block $\qs$}}

\STATE{Increase $\kqs\leftarrow \kqs+1$}

\STATE{\COMMENT {Compute $\vW\cqs_{\kqs}$ and 
$\vWt\cqs_{\kqs}$ 
from $\vd_{\ell\cqs}$, (c.f. \eqref{GS0} and \eqref{GSR}}}

\STATE{\COMMENT{Update the objective functional $\chi$ for 
the block $\qs$}}

\STATE{$\ds{\chi(\qs)=\left|\la \vWt\cqs_{\kqs}, \vf\cqs\ra 
\right|}$}

\STATE{Increase $j \leftarrow j+1$}

\ENDWHILE

\COMMENT {Calculation of coefficients and approximation}

\FOR  {$q=1 : Q$}  
\STATE{
$\vf^{{\rm{a}}} \{q\}= \vfq - \vRq;\,\,
c\cq_n=\la \vBq_n, \vfq\ra,\,n=1,\ldots,\kq$}

\ENDFOR

\COMMENT {Concatenation of the partition to produce the 
whole signal approximation}

$\vf^{{\rm{a}}} = \Join_{q=1}^Q \vf^{{\rm{a}}}\cq$
\end{algorithmic}
\end{algorithm}

\begin{algorithm}[!pht]
\setstretch{1.6}
\refstepcounter{myalg}
\caption{Atom Selection Procedure}
\label{alell}
\begin{algorithmic}
\STATE{\bf{Input:}}\, Residue $\vR\cqs$. Dictionary
$\ds{\D=\{\vd_n \in \R^{\Nq}\,; \|\vd_n\|=1 \}_{n=1}^M}$.
 Auxiliary sequence
$\{s\{\qs\}_n\}_{n=1}^{M}$ 
 and 
vector $\vWt\{\qs\}_{\kqs}$ for the upgrading.  Set of 
already selected atoms $\Gamma\cqs$.

\STATE{\bf{Output:}}\, Index $\ell\cqs$ 
corresponding to the selected atom
 $\vd_{\ell\cqs}$.\\

\COMMENT{Apply the Auxiliary Procedure, 
Algorithm \ref{alsn}, 
to upgrade the sequence $\{s\{\qs\}_n\}_{n=1}^{M}$ 
with respect to vector $\vWt\{\qs\}_{\kqs}$}

\COMMENT{Select the index, as below}

$\ds{{\ell{\cqs}=\argmax_{\substack{n=1,\ldots,M \\n \notin \Gamma \cqs}} \frac{\left| \la \vd_n , \vR\cqs \ra \right |}{(1-s\cqs_n)^{\frac{1}{2}}}}}$
\end{algorithmic}
\end{algorithm}

\begin{algorithm}[!pht]
\setstretch{1.3}
\refstepcounter{myalg}
\caption{Auxiliary Procedure}
\label{alsn}
\begin{algorithmic}

\STATE{\bf{Input:}}\, Dictionary
$\ds{\D=\{\vd_n \in \R^{\Nq}\,; \|\vd_n\|=1 \}_{n=1}^M}$.
 Auxiliary sequence
$\{s\{\qs\}_n\}_{n=1}^{M}$ and vector $\vWt\{\qs\}_{\kqs}$
for the upgrading.

\STATE{\bf{Output:}}\, Upgraded sequence 
$\{s\cqs_n\}_{n=1}^M$, with respect to
 $\vWt\{\qs\}_{\kqs}$.

\FOR {$n=1:M$}
\STATE{
$\ds{s\{\qs\}_n= s\{\qs\}_{n} + \left |\la \vd_n, \vWt\{\qs\}_{\kqs}\ra \right |^2}$}
\ENDFOR
\end{algorithmic}
\end{algorithm}

\subsection{Trigonometric dictionaries and 
sparse representation of music signals}

The viability of a sparse representation
for a given signal depends in large part  on
the dictionary choice. A possibility to 
 secure a suitable dictionary for approximating 
a signal by partitioning is through 
learning techniques \cite{KDME03,AEB06,RZE10,TF11}. 
Another possibility 
is to build the dictionary by merging
 basis, or dictionaries, containing atoms of
 different nature.  For instance,  a
 Cosine-Dirac dictionary 
enlarged by the incorporation of B-spline  
 dictionaries has been shown to be 
suitable for producing 
highly sparse high quality approximation of images 
\cite{RNB13}. Such dictionaries are very redundant.
In the case of melodic music, however, stunning sparsity 
levels may be achieved with relatively much less redundant 
 dictionaries. As will be illustrated here by 
numerical examples, it is the combination of  
a Redundant Discrete Cosine Dictionary (RDCD) $\Dc$
and a Redundant Discrete Sine Dictionary (RDSD) $\Ds$,
defined
below, which yields highly sparse representation of 
music signals when processed with the proposed techniques.
\begin{itemize}
\item
$\Dc=\{\frac{1}{w^{\rc}(n)}
 \cos{\frac{{\pi(2i-1)(n-1)}}{2M}},i=1,\ldots,\Nq\}_{n=1}^{M}.$
\item
$
\Ds=\{\frac{1}{w^{\rs}(n)}\sin{\frac{{\pi(2i-1)(n)}}{2M}},i=1,\ldots,\Nq\}_{n=1}^{M},$
\end{itemize}
where $w^{\rc}(n)$ and $w^{\rs}(n),\, n=1,\ldots,M$
are normalization factors as given by
$$w^{\rc}(n)=
\begin{cases}
\sqrt{\Nq} & \mbox{if} \quad n=1,\\
\sqrt{\frac{\Nq}{2} + \frac{ 
\sin(\frac{\pi (n-1)}{M}) \sin(\frac{2\pi(n-1)\Nq}{M})}
{2(1 - \cos(\frac{ 2 \pi(n-1)}{M}))}} & \mbox{if} \quad n\neq 1.
\end{cases}
$$
$$w^{\rs}(n)=
\begin{cases}
\sqrt{\Nq} & \mbox{if} \quad n=1,\\
\sqrt{\frac{\Nq}{2} - \frac{ 
\sin(\frac{\pi n}{M}) \sin(\frac{2\pi n \Nq}{M})}
{2(1 - \cos(\frac{ 2 \pi n}{M}))}} & \mbox{if} \quad n\neq 1.
\end{cases}
$$
For $M=\Nq$ each of the above dictionaries is an  
orthonormal basis, the Orthogonal Discrete Cosine Basis 
(ODCB)
and the Orthogonal Discrete
Sine  Basis (ODSB), henceforth to be denoted $\Bc$ and
$\Bs$ respectively. The joint
dictionary is an orthonormal
basis for $M=\frac{\Nq}{2}$, the Orthogonal Discrete 
Cosine-Sine Basis (ODCSB) 
to be indicated as $\Bcs$.  For $M>\Nq$,
 $\Dc$ is a RDCD
and $\Ds$ a RDSD. If $M> \frac{\Nq}{2}$
$\Dcs = \Dc \cup \Ds$ becomes a Redundant Discrete 
Cosine-Sine  Dictionary (RDCSD).

For the sake of discussing a fast calculation of inner
products
with trigonometric atoms,
 given a vector $\vy\in \R^M$, let's define
\be
\label{dft}
{\cal{F}}(\vy,n,M)= \sum_{j=1}^M y(j) e^{-\im  2 \pi \frac{(n-1)(j-1)}{M}},\quad n=1,\ldots,M.
\ee
When $M=\Nq$ \eqref{dft} is the Discrete Fourier Transform
of vector $\vy \in \R^{\Nq}$, which can be evaluated using
FFT. If $M>\Nq$ we can still calculate \eqref{dft} via
FFT by padding with zeros the vector  $\vy$.
Thus, \eqref{dft} can be used to calculate inner
products with the atoms in dictionaries $\Dc$ and
$\Ds$. Indeed,
\be
\label{dct}
\sum_{j=1}^{\Nq}\cos{\frac{{\pi(2j-1)(n-1)}}{2M}}
y(j)= \Re \left(e^{-\im \frac{\pi (n-1)}{2M}}
{\cal{F}}(\vy,n,2M)\right),\,
n=1,\ldots,M.
\ee
and
\be
\label{dst}
\sum_{j=1}^{\Nq}\sin{\frac{{\pi(2j-1)(n-1)}}{2M}}
y(j)= -\Im \left(e^{-\im \frac{\pi (n-1)}{2M}}
{\cal{F}}(\vy,n,2M)\right), \,
n=2,\ldots, M+1,
\ee
where $\Re(z)$ indicates the real part of $z$ and
$\Im(z)$ its imaginary part.

The computation  of inner products with 
trigonometric dictionaries via FFT is outlined
in Algorithm \ref{IP}.
For dictionaries $\Dc$ and $\Ds$ that 
 procedure reduces the complexity for calculating 
the inner 
products 
 from O$(\Nq M)$, per block, to O$(2M\log_2 2M)$, per block.
For dictionary $\Dcs$ the 
reduction is larger:
O$(M\log_2 M)$, because both the real and imaginary parts 
of the FFT are used. 
When  Algorithm 1 is applied with 
trigonometric dictionaries the pieces implemented by  
Algorithms \ref{alell} and \ref{alsn} should be 
replaced by Algorithms \ref{alellfft} and \ref{alsnfft},
given in Appendix A. In addition to speeding 
the calculations 
via FFT, Algorithms \ref{alellfft} and \ref{alsnfft} avoid 
 storing the dictionary. In the case of the dictionary 
$\Dcs$ it is assumed that both the Cosine and Sine 
components are of the same size and the first $\frac{M}{2}$
elements of the dictionary are Cosine atoms. 
\begin{algorithm}[!ht]
\refstepcounter{myalg}
\setstretch{1.6}
\begin{algorithmic}
\label{IP}
\caption{Computation of inner products with  a
trigonometric dictionary via FFT.\hspace{2cm}
IPTrgFFT procedure: $\vP=$IPTrgFFT($\vR, M,$ Case)}

\STATE{\bf{Input:}}$\,\vR \in \R^{\Nq}$, $M$, number of
elements in the dictionary, and Case (`Cos' or `Sin').\\

\STATE{\bf{Output:}} Vector $\vP\in\R^{M}$, 
with the inner products  between $\vRe$ and `Cos' or `Sin' 
dictionaries. 

\STATE \COMMENT{Computation of auxiliary vector
$\vt \in \C^{M}$ to compute $\vP$.}

\STATE{$\vt={\rm{FFT}}(\vR, M)$} \COMMENT{FFT with 
$(M-\Nq)$ zero padding}.

\STATE{Case `Cos'}

\STATE{$p(n)=\frac{1}{w^c(n)} 
\Re(\expn t(n)),\, n=1,\ldots,M$}\COMMENT{(c.f. \eqref{dct})}

\STATE{Case `Sin'}

\STATE{$p(n)=-\frac{1}{w^s(n)}\Im(\expn t(n)),\, n=1,\ldots,M$} \COMMENT{(c.f.\eqref{dst})}

\label{AlIP}
\end{algorithmic}
\end{algorithm}

{\bf{Numerical Example I}}\\
The signal approximated in this example is a piece of 
piano melody shown in Fig 1. 
It consists of $N=960512$ samples (20 secs)
divided into $Q=938$ blocks with $\Nq=1024$
samples each.
\begin{figure}[!ht]
\begin{center}
\includegraphics[width=9cm]{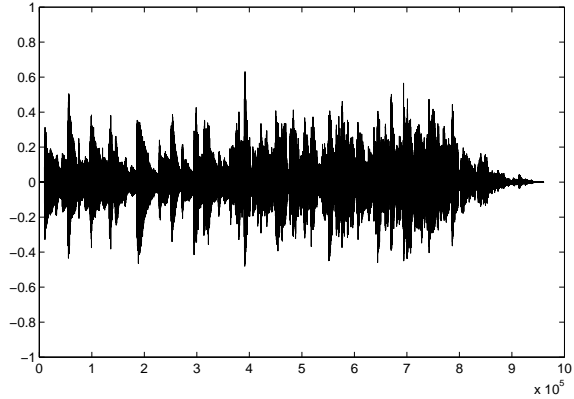}
\caption{{\small{Piano melody.
Credit: Julius O. Smith,
Center for Computer Research in Music and Acoustics (CCRMA),
Stanford University.}}}
\end{center}
\end{figure}
As sparsity measure for a signal representation we consider
the Sparsity Ratio (SR) defined as  
$\text{SR}= \frac{N}{K}$, where $K$ is the total number 
of coefficients in the signal representation.
As a measure of approximation 
quality we use the standard Signal to Noise Ratio
(SNR), 
$$\text{SNR}=10 \log_{10} \frac{\| \vf\|^2}{\|\vf - \vfa\|^2}=
10 \log_{10}\frac{\sum_{\substack{i=1\\q=1}}^{\Nq,Q} |f\{q\}(i)|^2}
{\sum_{\substack{i=1\\q=1}}^{\Nq,Q} |f\{q\}(i) -f^{\rm{a}}\{q\}(i)|^2}.$$
This numerical example aims at illustrating the 
following outcomes in relation to the signal in hand.
\begin{itemize}
\item [1)] The approximation power of all the 
orthogonal basis and dictionaries  defined above
remarkably improve if, instead of  applying OMP/OOMP
independently to approximate each block $\vfq$ up to
the same SNR, 
the HBW-OMP/OOMP approach is applied to 
match the sparsity of the whole signal.
\item [2)]For approximating the signal with the
 OMP/OOMP greedy strategies, the redundant 
dictionaries $\Dc$ and $\Dcs$
 perform significantly better than any of the 
 orthogonal basis 
$\Bc, \Bs$ and $\Bcs$. 
\end{itemize}
\begin{table}[!ht]
\label{tab1}
\begin{center}
\begin{tabular}{| l || r |r|| r|r || r|r || r|r ||}
\hline
Dict.&\multicolumn{2}{|c||}{OMP} & \multicolumn{2}{|c||}{HBW-OMP}
&\multicolumn{2}{|c||}{OOMP}&\multicolumn{2}{|c||}{HBW-OOMP}\\\hline \hline
 &SR&SNR&SR&SNR  &SR&SNR&SR&SNR\\ \hline \hline
$\Bc$&14.38&25.0& 14.38 & 35.19   &14.38&25.0& 14.38 & 35.19\\ \hline
$\Dct$&17.75&25.0& 17.75 &34.91&18.47&25.0& 18.47 &35.08 \\ \hline
$\Dcf$&19.39&25.0& 19.39  & 34.64 &19.80&25.0& 19.80& 34.95 \\ \hline \hline
$\Bs$&7.65&25.0& 7.65 & 28.31 &7.65&25.0&7.65 & 28.31\\ \hline
$\Dst$& 12.13 &25.0& 12.13 &29.39&12.63 &25.0&12.63 &29.39 \\ \hline
$\Dsf$& 13.17 &25.0&  13.17&29.34&13.42 &25.0&13.42 &29.33  \\ \hline \hline
$\Bcs$& 10.77&25.0&  10.77  & 30.09&10.77&25.0&10.77&30.09   \\ \hline
$\Dcst$& 20.45  &25.0&  20.45&35.81&22.08&25.0&22.08&35.56\\ \hline
$\Dcsf$ & 23.83 &25.0& 23.83& 35.56 &26.18&25.0&26.18&36.37\\ \hline \hline
\end{tabular}
\caption{\small{Comparison of the approximation quality
(SNR values) and sparsity (SR values) produced
 with trigonometric basis $\Bc, \Bs, \Bcs$ and
redundant trigonometric dictionaries,
$\Dct, \Dcf$,$\Dst, \Dsf, \Dcst$, and
$\Dcsf$.  The second column shows the
SR resulting when applying the block independent OMP
approach to achieve a SNR=25.0dB. The fifth
column demonstrates the significant gain in SNR
rendered by the HBW-OMP strategy for the same sparsity.
Further improvements of the same nature
 are shown in the last four columns
 corresponding to the approaches
OOMP and  HBW-OOMP.
}}
\end{center}
\end{table}
In oder to demonstrate 1) and 2) each of the blocks 
$\vfq$ is
approximated independently with OMP/OOMP, up to
SNR=25 dB. Redundant dictionaries,
with redundancy 2 and 4, are simply creating
by setting $M=2\Nq$ and $M=4\Nq$ in the definitions
of $\Dc$ and $\Ds$, with $\Nq=1024$. They will be
denoted as $\Dct$, $\Dst$ and
$\Dcf$, $\Dsf$, respectively. 
Approximations of the same quality are performed
with each of the orthonormal basis $\Bc, \Bs$ and $\Bcs$, 
and the mixed dictionaries $\Dcst=\Bc \cup \Bs$ and 
$\Dcsf=\Dct \cup \Dst$. 
The results are presented in Table 1.
As can be seen, the SR produced  
by the redundant dictionaries 
$\Dct, \Dcf, \Dcst$, and $\Dcsf$
is substantially larger 
than that corresponding to any of the orthogonal basis.  
Moreover, in all the cases the HBW-OMP/OOMP strategies 
improve notoriously (up to 11 dB) 
upon the OMP/OOMP approaches applied independently to
produce a uniform SNR in the approximation of 
each block.  

It is noticed that the highest sparsity
is attained by the dictionary $\Dcsf$ which, as already 
discussed, involves the most effective implementation 
via FFT. It is also seen 
that HBW-OOMP
over performs  HBW-OMP in terms of sparsity 
(the running times are similar). This will appear even more 
clearly in Table 2,  where the approximation 
of both methods are 
downgraded to produce a SNR of 25dB.
The next section discusses the HBW backward strategy
which allows for removing atoms from the signal
 approximation, in a stepwise optimized fashion. 
\subsection{Hierarchized Blockwise Backwards OOMP}
We extend here the Backward Optimized Orthogonal 
Matching Pursuit (BOOMP) strategy \cite{ARNS04} to select 
also the blocks from which the atoms are to be removed for 
downgrading the approximation of a whole signal.
 The BOOMP strategy
is stepwise optimal because it 
minimizes, at each step, the norm of the error resulting by 
 downgrading, by one atom, an atomic decomposition. This 
result readily follows from the
 recursive equations for 
modifying vectors
$\{\vBq_n\}_{n=1}^{\kq}$  
to account for the elimination 
of, say the $j$-th atom, from the set 
$\{\vd_{\ellq_n}\}_{n=1}^{\kq}$. For each $q$, the reduced  
set of vectors
$\{\vBq_{n}\}_{\substack{n=1\\ n \neq j}}^{\kq}$ 
spanning  the reduced 
subspace $\V_{\kq/j}^q=\Spann \{\vd_{\ellq_n}\}_{\substack {n=1\\ n \neq j}}^{\kq}$ 
can be quickly obtained through the adaptive 
backward equations \cite{ARNS04,RN04}
\be
\label{bacB}
\vB\cq_n \leftarrow \vB\cq_n - \vB\cq_j \frac{\la \vB\cq_n, \vB\cq_j \ra}{\|\vB\cq_j\|^2},\quad n=1,\ldots,j-1,j+1,\ldots,\kq.
\ee
Consequently, the coefficients of the atomic decomposition 
corresponding to the block $q$, from which the atom 
$\vd_{\ellq_j}$ is 
 taken away, are modified as 
\be
\label{bacc}
c\cq_n \leftarrow  c\cq_n -c\cq_j \frac{\la \vB\cq_j, \vB\cq_n \ra}{\|\vB\cq_j\|^2},\quad
n=1,\ldots,j-1,j+1,\ldots,\kq.
\ee
For a fixed value of $q$, the BOOMP criterion  
removes the 
coefficient $c\cq_{\jd}$ such that \cite{ARNS04}
\be
\jd\cq=
\operatorname*{arg\,min}_{\substack{j=1,\ldots,\kq}}
 \frac{|c\cq_j|}{\|\vB\cq_j\|^2}.
\label{boompj}
\ee
The equivalent criterion applies to the selection of 
the $\qd$-block for the downgrading of the whole
approximation. The proof of the next proposition 
parallels the proof of \eqref{boompj} given in \cite{ARNS04}.
\begin{proposition}
Assume that the approximation of a signal is given as
$\vfa= \Join_{q=1}^Q \vfaq$, 
where
$\vfaq= \op_{\V_{\kq}^q}\vfq$.
Let $\jd\cq,\, q=1,\ldots,Q$ be the indices
 which satisfies  
\eqref{boompj} for each $q$.  The block $\qd$,
from where the atom
 $\vd_{\ell_{\cqd_{\jd\cqd}}}$
is to be removed, 
 in order to leave an approximation $\vfad$ 
such that
 $\|\vfa - \vfad\|$ takes its minimum value,
 satisfies the condition:
\be
\label{boompq}
\qd=
\operatorname*{arg\,min}_{q=1,\ldots,Q} 
\frac{|c\cq_{\jd\cq}|}{\|\vB\cq_{\jd\cq}\|^2}.
\ee
\end{proposition}
\begin{proof}
Since at each step only one
atom is removed from 
 a particular block, say the $q$-th one, it holds that
$\|\vfa - \vfad\|= \|\vfa\cq - \vfad\cq\|.$
  The identical steps as in \cite{ARNS04} 
lead to the expression
\ben
\vfa\cq - \vfad\cq= \vd_{\ell_{\cq_{\jd\cq}}}\frac{\la \vBq_{\jd\cq},\vfq\ra}{\|\vBq_{\jd\cq}\|^2}.
\een
Thus $\ds{\|\vfa\cq - \vfad\cq\|^2= \frac{\left |\la \vBq_{\jd\cq}, \vfq\ra \right|^2}{{\|\vBq_{\jd\cq}\|^4}}}$.
Since $\la \vBq_{\jd\cq}, \vfq\ra= c\cq_{\jd\cq}$ it  
  follows that 
$\|\vfa - \vfad\|$  
is minimized by removing the atom  $\jd\cqd$
corresponding to the block $\qd$ satisfying 
\eqref{boompq}. 
\end{proof}
The HBW-BOOMP algorithm for downgrading a given atomic 
decomposition of a signal is 
 presented in Algorithm 5.

\begin{algorithm}[!htp]
\label{alhbwoomp}
\setstretch{1.5}
\refstepcounter{myalg}
\caption{HBW-BOOMP}
\label{alhbwboomp}
\begin{algorithmic}
\STATE{\bf{Input:}}\,
Biorthogonal 
sets $\ds{\{\vBq_n\}_{n=1}^{\kq}},\,\,q=1,\ldots,Q$.
Coefficients in the approximation of 
each block $\ds{\{c\cq_n\}_{n=1}^{\kq}},\, q=1,\ldots,Q$.
Sets of indices  in the decomposition 
of each block $\Gamma\cq,\, q=1,\ldots,Q$. 
Number of total atoms for the downgraded approximation, 
$K_d$.

\STATE{{\bf{Output:}}\, 
Downgraded biorthogonal sets.
Coefficients of the downgraded atomic decompositions for 
each block. Indices of the  remaining atoms.}\\

\STATE {\COMMENT {Initialization}}

$K=0$;
\FOR {$q=1:Q$}

\STATE{
$\kq =|\Gamma\cq|$\,\,\COMMENT{Number of 
elements in $\Gamma\cq$} 

$K=K+\kq$

\COMMENT{Select the index of the potential first atom to 
be removed in each block, as below} 

$\ds{\jdq= \argmin_{n=1,\ldots,\kq}\frac{|c\cq_n|}{\|\vBq_n\|^2}}$ 

$\ds{\chi_{\downarrow}(q)=\frac{|c\cq_{\jdq}|}{\|\vBq_{\jdq}\|^2}}$
\,\,\COMMENT{Store the minimum}}

\ENDFOR

$i=K$

\WHILE {$i > K_{\rm{d}}$} 
\STATE{
\COMMENT{Select the block to be downgraded, as below}  

{$\ds{\qd= \argmin_{q=1,\ldots,Q} \chi_{\downarrow}(q)}$}

\COMMENT{Apply backward biorthogonalization  
(c.f. \eqref{bacB}) to downgrade the biorthogonal set 
corresponding to block $\qd$, with respect 
to $\vB\cqd_{\jd\cqd}$ and shift indices}

\COMMENT{Update the coefficients corresponding to 
block $\qd$ (c.f. \ref{bacc}) and shift indices}

{Downgrade the set $\Gamma{\cqd} \leftarrow \Gamma{\cqd}/\ell\cqd_{\jdqd}$ and shift indices}

Decrease $\kq \leftarrow \kq-1$ 

\COMMENT{Select the new potential atom to be 
removed from block $\qd$, as below}

{$\ds{\jdqd= \argmin_{n=1,\ldots,\kq}
\frac{|c\cqd_n|}{\|\vB\cqd_n\|^2}}$}

\COMMENT{Update the objective functional $\chi_{\downarrow}$
for the block $\qd$, as below}

$\ds{\chi_{\downarrow}(\qd)=\frac{|c\cqd_{\jdqd}|}{\|\vB\cqd_{\jdqd}\|^2}}$

{Decrease $i \leftarrow i-1$}
}
\ENDWHILE
\end{algorithmic}
\end{algorithm}

\subsection*{Numerical Example II}
In order to illustrate
the backward approach, we change
 the  information in the last four columns of Table 1 to
have the sparsity of all the approximations corresponding to
 SNR=25dB.
For this we downgrade the HBW-OMP/OOMP
approximations to degrade the previous quality.
\begin{table}[!ht]
\begin{center}
\begin{tabular}{| l || r |c|| r|c || r |c|| r|c ||}
\hline
Dict.&\multicolumn{2}{|c||}{OMP} & \multicolumn{2}{|c||}{HBW-BOOMP} & \multicolumn{2}{|c||}{OOMP} & \multicolumn{2}{|c||}{HBW-BOOMP}
 \\\hline \hline
 &SR&SNR&SR&SNR  &SR&SNR&SR&SNR\\ \hline \hline
$\Bc$&14.38&25.0& 25.56  & 25.0     &  14.38&25.0&25.56  &25.0 \\ \hline
$\Dct$&17.60&25.0&  34.05 &25.0  &  18.47    &25.0&36.34  &25.0 \\\hline
$\Dcf$&19.25&25.0&  37.01 & 25.0 &  19.80    &25.0&39.06  &25.0  \\ \hline \hline
$\Bs$&7.65&25.0&  13.67   & 25.0   &  7.65     &25.0&13.67  &25.0  \\ \hline
$\Dst$& 12.03 &25.0&  25.74 & 25.0&  12.63     &25.0&27.56  &25.0  \\ \hline
$\Dsf$& 13.11 &25.0& 28.18&  25.0&   13.42     &25.0&29.59  &25.0  \\   \hline \hline
$\Bcs$& 10.77&25.0&  19.94  &  25.0&  10.77     &25.0&19.94  &25.0 \\ \hline
$\Dcst$& 20.21  &25.0&  42.43& 25.0&   22.08    &25.0&48.48  &25.0 \\ \hline
$\Dcsf$ &  23.53 &25.0& 50.19& 25.0&   26.18    &25.0&59.68  &25.0 \\ \hline \hline
\end{tabular}
\label{tab2}
\caption{\small{Comparison of sparsity (SR values) for
the same SNR.
The HBW-BOOMP results are obtained by degrading the
HBW-OMP/OOMP approximation in Table 1 to SNR=25dB.}}
\end{center}
\end{table}
As seen in Table 2, the result is that
the sparsity increases drastically. The SR
improves up to $128\%$ (for the $\Dcsf$ dictionary) 
 in comparison with the standard application of 
the same approaches without ranking the blocks.
 Notice also that the best SR result 
(for dictionary $\Dcsf$), is $133\%$ higher than the best SR result for an 
orthogonal basis (the $\Bc$ basis). 
\section{Refinement by Swaps}
The suitability of the HBW strategy for approximating
the class of signals we are considering is a notable fact.
However, the possibility of approximating every element of
the partition  completely independent of the others has
the convenient feature of leaving
room for straightforward parallelization. 
 This is particularly important for implementations 
on Graphics Processing Units
(GPU), for instance. Accordingly, an alternative worth
exploring is to maintain the block independent
approximation but allowing for possible eventual
HBW refinements.

Assuming that the approximation \eqref{atoq} of 
each block in a signal partition is known, the goal 
now is to improve the signal approximation
maintaining unaltered the total 
number of atoms $K$.
The proposed refinement consists of
 movement of atoms controlled by the following 
 instructions:
\begin{itemize} 
\item [i)]Use criterion \eqref{boompq} to remove one  
atom from the block $\qd$. Let us call this block 
a `donor' block and indicate it 
with the index $q_{{\rm{d}}}$. 

\item [ii)]Use criterion \eqref{hbwoomp} to incorporate 
one atom in the approximation of the block $q^\star$.
 Let us call this block a `receiver' and 
indicate it with the index $q_{{\rm{r}}}$. 

\item [iii)]Denote by $\delta_{{\rm{d}}}$
 the square norm of the error
 introduced at step i)by downgrading  
the approximation of the donor block $q_{{\rm{d}}}$. 
Denote by $\delta_{{\rm{r}}}$ the square norm of the gain in improving the approximation of the receiver $q_{{\rm{r}}}$.
Proceed according to the following rule: 

If  $\delta_{{\rm r}} > \delta_{{\rm d}}$ accept the 
change and repeat steps i) and ii).  Otherwise stop.
\end{itemize}
The above procedure is termed HBW-SbR-OMP/OOMP 
according to whether the selection of the 
atom at stage ii) is realized through the 
 OMP (c.f. \eqref{omp}) or OOMP (c.f. \eqref{oomp2})
criterion.\\
Notice that the HBW-SbR-OMP/OOMP methods, implemented as
proposed above, possess the desirable feature of
acting {\em{only}} if the approximation can be
improved by the proposed swap. 
The possible situation for which  $q_{{\rm{r}}}=
q_{{\rm{d}}}$ corresponds to a swap in the original 
SbR-OMP/OOMP approach \cite{ARN06}, 
 which may take place here as a particular case.

{\bf{Implementation Remark:}}
For successive implementation of the 
steps i) and ii) 
one needs: 
to downgrade the vectors $\vB\cqd_n,\, n=1,\ldots \kq$ 
at step i) and to upgrade these vectors at step ii). 
In order to allow for step ii) 
the orthogonal set  $\left\{\vW\cqd_n \right\}_{n=1}^{\kq}$
has to be downgraded as well. This can be realized 
effectively by the plane rotation 
approach \cite{lsbook,Stw95}.

\subsection*{Numerical Example III}
We illustrate the HBW-SbR-OMP/OOMP approaches by applying them to the signal in Fig. 2, which is a flute  exercise
consisting of $N=96256$ samples
divided into $Q=94$ blocks with $\Nq=1024$
samples each.\\
\begin{figure}[!ht]
\begin{center}
\includegraphics[width=9cm]{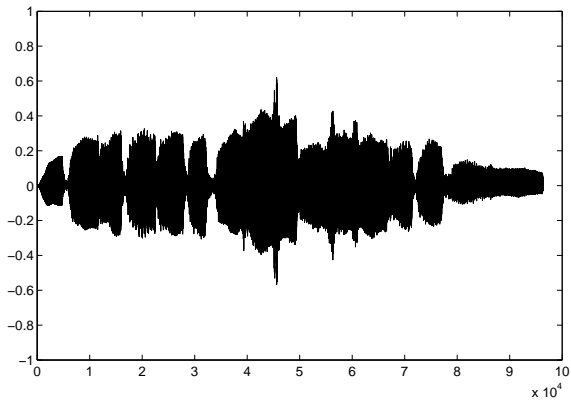}
\vspace{-0.5cm}
\caption{{\small{Flute exercise:
$N=96256$ samples.}}}
\end{center}
\end{figure}
Table 3 and 4 show the improvement in quality obtained 
by applying the HBW-SbR-OMP/OOMP approaches on outputs 
of the OMP/OOMP approximations. 
Notice that in most cases 
the SNR results 
are practically equivalent to those obtained by the 
HBW-OMP/OOMP approaches.
\begin{table}[!ht]
\label{tab3}
\begin{center}
\begin{tabular}{| l || r |r|r |c |c||}
\hline
Dict.& SR & OMP & Swaps & HBW-SbR-OMP & HBW-OMP \\ \hline \hline
$\Bc$&10.13&25.0& 1426& 26.80& 26.79\\ \hline
$\Dct$&14.26 &25.0& 1013 &26.69 & 26.67\\ \hline
$\Dcf$&15.58&25.0& 945 & 26.63 & 26.61\\ \hline \hline
$\Bs$&7.27&25.0& 2612 & 26.72& 26.73 \\ \hline
$\Dst$& 12.10 &25.0& 1397 & 27.24& 27.27\\ \hline
$\Dsf$& 13.16 &25.0& 1435&  27.27 & 27.25\\ \hline \hline
$\Bcs$& 8.08&25.0& 1822 &  26.65& 26.65   \\ \hline
$\Dcst$&  29.00 &25.0&  636 & 27.38& 27.37\\ \hline
$\Dcsf$ &  35.17 &25.0& 518& 27.36& 27.36 \\ \hline \hline
\end{tabular}
\caption{\small{Comparison of quality (SRN values) for
a fixed sparsity: that corresponding to SRN=25 with the
OMP approach. The forth column shows the number of
swaps and the fifth columns the SRN achieved  by those
swaps through the HBW-SbR-OMP refinement  to the
outputs yielding the second column.
For further comparison the last column shows the
results corresponding to the HBW-OMP approach.}}
\end{center}
\end{table}
\begin{table}[!ht]
\label{tab3}
\begin{center}
\begin{tabular}{| l || r |r|r |c |c||}
\hline
Dict.& SR & OOMP & Swaps & HBW-SbR-OOMP & HBW-OOMP \\ \hline \hline
$\Bc$&10.13&25.0& 1426& 26.80& 26.79\\ \hline
$\Dct$&14.83 &25.0& 865 &26.66 & 26.68\\ \hline
$\Dcf$&15.80&25.0& 929 & 26.60 & 26.61\\ \hline \hline
$\Bs$&7.27&25.0& 2612 & 26.72& 26.73 \\ \hline
$\Dst$&12.58 &25.0& 1531 & 27.14& 27.27\\ \hline
$\Dsf$&13.34 &25.0& 1445&  27.27 & 27.26\\ \hline \hline
$\Bcs$&8.08&25.0& 1822 &  26.65& 26.65   \\ \hline
$\Dcst$&33.47 &25.0&  556 & 27.64& 27.63\\ \hline
$\Dcsf$&42.05 &25.0& 228& 27.10& 27.69 \\ \hline \hline
\end{tabular}
\caption{\small{Same description as in Table 3 but
for the OOMP, HBW-SbR-OOMP, and HBW-OOMP approaches.}}
\end{center}
\end{table}

\section{Processing of large signals}
\label{PLS}
As already discussed, the additional computational
 complexity introduced 
 by the HBW raking of $Q$ blocks is 
 only O($Q$) (over the complexity of a 
 standard pursuit algorithm implemented without 
ranking the blocks). Storage does 
become more demanding though.
The need for saving
 vectors \eqref{vB} and \eqref{GS0}, for instance,
implies to have to store $2K$ vectors of size $\Nq$.
This requirement restricts the amount of
blocks to be processed with a small system such as 
a standard laptop. Hence, 
we outline here  a processing scheme which 
makes it possible the application of HBW techniques 
on very large signals. 
For this the whole signal needs to be  
divided into large segments of convenient size, 
say $\Ns$. 
More specifically, a number of, say $P$ blocks, of size 
$\Nq$
are grouped together to form $S$ larger segments
$\vg\cs = \Join_{q=(s-1)P+1}^{sP} 
\vfq,\,s=1,\ldots,S$,  where 
it is assumed that 
$\Nq P= \Ns$ and $S\Ns= \Nq Q$.
 Each segment $\vg\cs \in \R^{\Ns}$ is 
approximated with a HBW strategy as an entire individual signal.
The approximation of the  original signal is 
assembled at the very end of the process  by the 
operation
${\ds{\vf^{\rm{a}} = \Join_{s=1}^{S}{\vg^{\rm{a}}\cs}}}$.
In other words, the processing of large signals is 
realized by chopping the signal into segments and 
processing each segment independently of the others. 
Nevertheless, the question as to how to 
set the sparsity 
constraint for each segment needs further 
consideration. One could, of course, 
require the same sparsity in every segment, 
unless information advising otherwise were available. 
While uniform sparsity guarantees the same sparsity
on the whole signal, in particular 
cases, where the nature of the signal changes over 
its range of definition, it would not render the best 
approximation. In order to avoid encountering this
type of situation we introduce a preliminary step:
{\em{the randomization of all the small blocks in the 
signal partition}}. The implementation is carried out
 as follows: 
\begin{itemize}
\item [i)]
Given a signal $\vf$ split it into $Q$ blocks  
of size $\Nq$. Apply an invertible random permutation 
 $\Pi$ to scramble the block's location (the bottom 
graph in Figure 3 provides 
a visual illustration of how the 
signal in the top graph looks after this step).
Let's denote 
the re-arranged signal by
$\vft= \Join_{q=1}^Q \vft\cq$. 
\item [ii)]
Group every $P$ of the blocks $\vft\cq,\,q=1,\ldots,Q$  to 
have  $S$ segments 
$\vgt\cs = \Join_{q=(s-1)P+1}^{sP} 
\vft\cq,\,s=1,\ldots,S=\frac{Q}{P}$.
\item [iii)]
Approximate each of the segments $\vgt\cs$ independently 
of the others, as it each segment were an 
entire signal, i.e., 
obtain $\vgt^{\rm{a}}\cs= \Join_{q=(s-1)P+1}^{sP} 
\vft^{\rm{a}}\cq,\,s=1,\ldots,S$.
\item [iv)]
Assemble the segments to have the approximation of the 
whole signal as $\vft^{\rm{a}}=\Join_{s=1}^S \vgt^{\rm{a}}\cs$.
\item [v)]
Reverse the permutation of the block's location
to obtain, from $\vft^{\rm{a}}$, the approximation 
of the original signal $\vfa$.
\end{itemize}
\begin{figure}[h!]
\begin{center}
\includegraphics[width=9cm]{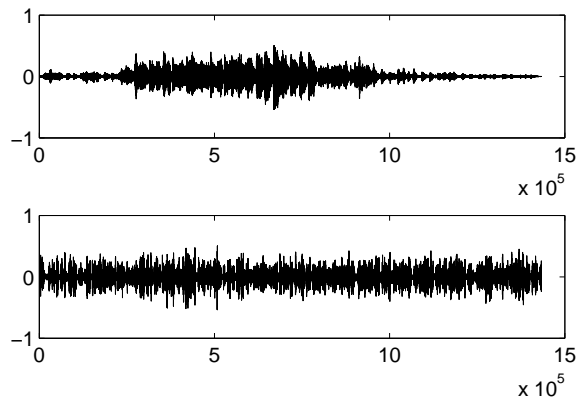}
\caption{{\small{The top graph shows a piece (1433600 samples) of Piazzola music interpreted by Gidon Kremmer and
orchestra. The bottom graph shows the
resulting signal after the randomization of the
blocks.}}}
\end{center}
\end{figure}
{\bf{Numerical Example IV}}\\
The signal to be approximated is a 32.5 secs
(1433600 samples) piece of
Piazzola music, shown in the top graph of Figure 3.  It is
partitioned into $Q=1400$ blocks
of 1024 points each. After randomization (bottom graph of
the same figure)
the blocks are grouped into
50 segments, each of which is independently
approximated by the HBW-OMP/OOMP approach to produce
a SR=11.53.
The resulting SNR of the approximated signal
is 29.12~dB with HBW-OMP and 30.15~dB with HBW-OOMP. 
Equivalent approximation quality 
is obtained with 
other realizations of the random process. For 
appreciation of this result the original signal 
was also approximated by HBW-OMP/OOMP, 
but without segmentation. 
In that case the quality corresponding to SR=11.53 is only
 $1.9\%$ higher (SNR=29.66~dB with HBW-OMP and SNR=30.73~dB 
with HBW-OOMP). The conclusion is that, even if the 
largest the segments the less the distortion 
for the same sparsity, applying HBW-OMP/OOMP on segments 
of a large signal may be still significantly 
more advantageous
 than the independent approximation of the blocks. Indeed, 
in this example the
 latter yields a SNR of 25~dB with OMP and 26.1~dB 
with OOMP.
\section{Conclusions} 
Cooperative greedy pursuit strategies 
for approximating a signal partition 
subjected to a global constraint on sparsity have been 
considered. The cooperation between 
partition units was realized in several ways:
i)By ranking the partition units for their sequential
stepwise approximation (HBW-OMP/OOMP) ii)By ranking the partition units for stepwise downgrading of the whole 
approximation (HBW-BOOMP) iii)By allowing for 
  downgrading of some partition units (donors)  and 
upgrading of another partitions units (receivers), 
to improve the approximation
 quality (HBW-SbR-OMP/OOMP). Comparisons of the OMP and OOMP 
criteria indicate that for the type of dictionaries 
and signals considered here, the OOMP criterion 
renders improvements of sparsity up to 20$\%$ 
in relation to the equivalent strategy involving 
OMP.\\
The HBW algorithms maintain the complexity comparable 
with the
totally independent approximation of each partition unit.
The examples presented here, using 
trigonometric dictionaries for achieving  high quality 
approximation of musics 
signals, provide a clear illustration of the 
benefits arising by the proposed strategies. Indeed, for
a piano melody 
the increment in sparsity is 128$\%$  higher than 
the sparsity produced by approximating each partition 
unit up to same quality. Comparison of trigonometric 
dictionaries results, against results 
obtained with trigonometric basis, give an improvement of 
equivalent order: 133$\%$. The mixed Cosine-Sine dictionary 
has yielded the best sparsity results, not
only for the test signals included in this paper but 
for a significant  number of
other signals, all in the category of melodic music.\\
Because the proposed HBW implementation of pursuit 
strategies 
requires storage of order $\Nq K$ ($\Nq$ being  the 
length of each block in the partition and $K$ 
the number of total atoms to approximate the whole signals)  
a simple procedure for processing large signal was proposed:
Previous to dividing the signal 
into independent segments, each of which to be approximated 
with the identical constraint of sparsity, 
the randomization of the 
partition units is advised.

\newpage
\appendix
\renewcommand{\theequation}{A.\arabic{equation}}
\setcounter{equation}{0}
\appendix{\bf{Appendix A}}
\begin{algorithm}[!pht]
\setstretch{1.8}
\refstepcounter{myalg}
\caption{Atom Selection Procedure via FFT}
\label{alellfft}
\begin{algorithmic}
\STATE{\bf{Input:}}\, Residue $\vR\cqs$. Dictionary 
Case (`Cos', `Sin' or mixed `Cos-Sin') and 
number of elements, $M$.
Sequence 
$\{s\{\qs\}_n\}_{n=1}^{M}$ if Case is `Cos' or `Sin', 
and $\{s^{\rc}\{\qs\}_n\}_{n=1}^{M}, \{s^{\rs}\{\qs\}_n\}_{n=1}^{M}$ if Case is `Cos-Sin'. Vector $\vWt\{\qs\}_{\kqs}$ to 
upgrade the sequences. Set of   
already selected atoms $\Gamma\cqs$ (sets 
 $\Gamma^{\rc}\cqs$  and $\Gamma^{\rs}\cqs$ with the atoms 
in each component if Case is `Cos-Sin').

\STATE{\bf{Output:}}\, Selected index $\ell\cqs$ 
for a potential new atom for block $\cqs$. 

\STATE{Cases single `Cos' or `Sin'}

\COMMENT{Call IPTrgFFT procedure, Algorithm \ref{AlIP},
to calculate inner products, and apply Algorithm 
\ref{alsnfft} to upgrade the sequence $\{s\{\qs\}_n\}_{n=1}^{M}$ with respect to vector $\vWt\{\qs\}_{\kqs}$}

$\vP$=IPTrgFFT($\vR\cqs, 2 M,$ Case),

\COMMENT{Select the index for a new 
potential atom for the block $\qs$, as below.}

\STATE{$\ds{{\ell{\cqs}=\argmax_{\substack{n=1,\ldots,M \\n \notin \Gamma \cqs}} \frac{|p\cqs(n)|}{(1-s\cqs_n)^{\frac{1}{2}}}}}$
}

\STATE{Case `Cos-Sin'}

\COMMENT{Same precedure as for the previous case but 
for both `Cos' and `Sin' cases},

\STATE{$\vP^{\rc}\cqs$=IPTrgFFT($\vRe\cqs,M$ `Cos');\quad
$\vP^{\rs}\cqs$=IPTrgFFT($\vRe\cqs,M,$ `Sin')}

\COMMENT{Apply Algorithm \ref{alsnfft} to upgrade the sequences $\{s^{\rc}\{\qs\}_n\}_{n=1}^{M}$ and $\{s^{\rs}\{\qs\}_n\}_{n=1}^{M}$ and selec the index}

Set $M \leftarrow \frac{M}{2}$

\STATE{$\ds{\ell^{\rc}{\cqs}=\argmax_{\substack{n=1,\ldots,M \\n \notin \Gamma^{\rc} \cqs}} \frac{\left|p^{\rc}\cqs(n)\right|}{(1-s^{\rc}\cqs_n)^{\frac{1}{2}}}}$;\quad
{$\ds{\ell^{\rs}{\cqs}=\argmax_{\substack{n=1,\ldots,M \\n \notin \Gamma^{\rs}\cqs}} \frac{\left|p^{\rs}\cqs(n)\right |}{(1-s^{\rs}\cqs_n)^{\frac{1}{2}}}}$}}

\COMMENT{Evaluate the maximum value, as below}

\STATE{${\ds{\mu^{\rc}{\cqs}=\frac{\left|p^{\rc}\cqs(\ell^{\rc}{\cqs})\right|}{(1-s^{\rc}\cqs_{\ell^{\rc}\cqs})^{\frac{1}{2}}}}}$};\quad
{${\ds{\mu^{\rs}{\cqs}=\frac{\left|p^{\rs}\cqs(\ell^{\rs}{\cqs})\right|}{(1-s^{\rs}\cqs_{\ell^{\rs}\cqs})^{\frac{1}{2}}}}}$}


\STATE{$\mu\cqs=\max(\mu^{\rc}\cqs,\mu^{\rs}\cqs)$}
\IF {$\mu= \mu^{\rc}$} 
\STATE{$\ell{\cqs}=\ell^{\rc}{\cqs}$}
\ELSE 
\STATE{$\ell{\cqs}=\ell^{\rs}{\cqs} + M$}
\ENDIF
\end{algorithmic}
\end{algorithm}

\begin{algorithm}[!pht]
\setstretch{1.7}
\refstepcounter{myalg}
\caption{Auxiliary Procedure via FFT}
\label{alsnfft}
\begin{algorithmic}

\STATE{\bf{Input:}}\, Vector $\vWt\{\qs\}_{\kqs}$. 
 Dictionary Case: `Cos', `Sin' 
or `Cos-Sin', and number of elements $M$. Sequence 
$\{s\{\qs\}_n\}_{n=1}^M$,
for Cases `Cos' and `Sin', or $\{s^{\rc}\cqs_n\}_{n=1}^M$, 
and $\{s^{\rs}\cqs_n\}_{n=1}^M$, 
for Case `Cos-Sin'. 

\STATE{\bf{Output:}}\, Upgraded sequence,  
$\{s\cqs_n\}_{n=1}^M$, for Cases `Cos' and `Sin', 
or $\{s^{\rc}\cqs_n\}_{n=1}^M$ and $\{s^{\rs}\cqs_n\}_{n=1}^M$ for Case `Cos-Sin'.

\STATE{Case `Cos' or `Sin'}

\COMMENT{Call IPTrgFFT procedure, Algorithm \ref{AlIP}, 
to calculate inner products with vector 
$\vWt\{\qs\}_{\kqs}$}

$\vP\cqs$=IPTrgFFT($\vWt\{\qs\}_{\kqs}, 2 M,$ Case)

\COMMENT{Upgrade the sequence, as below}
\FOR {$n=1:M$}
\STATE{$\ds{s\{\qs\}_n= s\{\qs\}_{n} + \left | p\cqs(n) \right|^2}$}
\ENDFOR

\STATE{Case `Cos-Sin'}

\COMMENT{Call IPTrgFFT procedure, Algorithm \ref{AlIP},
to calculate the inner products between 
vector $\vWt\{\qs\}_{\kqs}$ and each of the
 dictionary components: Cases `Cos' and `Sin'}

\STATE{$\vP^c$=IPTrgFFT($\vWt\{\qs\}_{\kqs}, M,$ `Cos')}\\

\STATE{$\vP^s$=IPTrgFFT($\vWt\{\qs\}_{\kqs}, M,$ `Sin')}\\

\COMMENT{Upgrade the sequences, as below}

\FOR {$n=1:\frac{M}{2}$}

\STATE{$\ds{s^{\rc}\{\qs\}_n= s^{\rc}\{\qs\}_{n} + \left |p^{\rc}
\cqs(n)\right|^2}$}

\STATE{$\ds{s^{\rs}\{\qs\}_n= s^{\rs}\{\qs\}_{n} + \left 
|p^{\rs}\cqs(n)\right|^2}$}
\ENDFOR
\end{algorithmic}
\end{algorithm}

\bibliographystyle{elsart-num}
\bibliography{revbib}
\end{document}